\documentclass[journal,twoside]{IEEEtran}
\usepackage{cite}
\ifCLASSINFOpdf
   \usepackage[pdftex]{graphicx}
\else
   \usepackage[dvips]{graphicx}
\fi
\usepackage{epstopdf}
\usepackage[cmex10]{amsmath}
\usepackage{amsfonts,amssymb,amsthm}
\usepackage{empheq}
\usepackage{multirow}
\usepackage{algorithmic}
\usepackage{array}
\usepackage[tight,footnotesize]{subfigure}
\usepackage{xfrac}
\usepackage{url}
\allowdisplaybreaks

\usepackage[utf8]{inputenc} % for linux

\usepackage{amssymb,amsbsy,latexsym}

\usepackage{tikz}
\usetikzlibrary{shapes,arrows,calc,matrix}

\usepackage{color}

\newcommand{\bb}[1]{\mathbf{#1}}
\newcommand{\bs}[1]{\pmb{#1}}

\newcommand{\low}[1]{_{\scriptscriptstyle #1}}
\newcommand{\high}[1]{^{\scriptscriptstyle #1}}
\newcommand{\mat}[1]{\bb{#1}}
\newcommand{\gmat}[1]{\bs{#1}}

\newcommand{\bvec}[1]{\bb{#1}}
\newcommand{\gvec}[1]{\bs{#1}}
\newcommand{\blkmat}[1]{\begin{bmatrix} #1 \end{bmatrix}}

\newcommand{\ctr}[1]{\bvec{u}\low{#1}}

\newcommand{\om}[1]{\gvec{\omega}\low{#1}}

\newcommand{\dom}[1]{\dot{\gvec{\omega}}\low{#1}}

\newcommand{\lbd}[1]{\gvec{\lambda}\low{#1}}

\newcommand{\dlbd}[1]{\dot{\gvec{\lambda}}\low{#1}}

\newcommand{\ddlbd}[1]{\ddot{\gvec{\lambda}}\low{#1}}

\newcommand{\zero}[0]{\mat{0}}

\newcommand{\zeros}[2]{\zero_{{#1}\times{#2}}}

\newcommand{\eyefunc}[0]{\mat{I}}
\newcommand{\eye}[1]{\eyefunc_{#1}}

\newcommand{\set}[1]{\mathcal{#1}}

\newcommand{\diag}[1]{{\rm diag}(#1)}
\newcommand{\norm}[1]{\left\|#1\right\|}

\newcommand{\normH}[2]{\norm{#2}_{\set{H}_{#1}} }

\newcommand{\tr}[1]{{\rm tr}\left(#1\right)}

\newcommand{\convhull}[1]{ {\rm co}(#1) }

\newcommand{\matskewfunc}[0]{\mat{S}}

\newcommand{\matskew}[1]{\matskewfunc(#1)}

\newcommand{\R}[0]{\mathbb{R}}

\newtheorem{thm}{Theorem}
\newtheorem{prop}[thm]{Proposition}

\newtheorem{define}{Definition}
\DeclareGraphicsExtensions{.eps}
\graphicspath{{figures/}}

\begin{document}

\title{LiDAR-based Control of Autonomous Rotorcraft for the Inspection of Pier-like Structures: Proofs}

\author{Bruno~J.~Guerreiro,
        Carlos~Silvestre,
        Rita~Cunha,
        David~Cabecinhas %
    \thanks{This work was partially funded by the Macau Science and Technology Development Fund (FDCT), grant FDCT/\-048/\-2014/\-A1, by project MYRG2015-00127-FST of the Univ. of Macau, by the European Union’s Horizon 2020 research and innovation programme under grant agreement No 731667 (MULTIDRONE), and by the Portuguese Foundation for Science and Technology (FCT), under projects LARSyS UID/\-EEA/\-50009/\-2013 and LOTUS PTDC/\-EEI-AUT/\-5048/\-2014. 
	The work of B. Guerreiro and R. Cunha were respectively supported by the FCT Post-doc Grant SFRH/\-BPD/\-110416/\-2015 and FCT Investigator Programme IF/00921/2013. 
	This publication reflects the authors’ views only. The European Commission and other funding institutions are not responsible for any use that may be made of the information it contains. }
    \thanks{B. Guerreiro and R. Cunha are with the Institute for Systems and Robotics (ISR-Lisbon), LARSyS, Instituto Superior Técnico (IST), Universidade de Lisboa, Av. Rovisco Pais 1, 1049 Lisboa, Portugal. ({\tt bguerreiro@isr.tecnico.ulisboa.pt} and {\tt rita@isr.tecnico.ulisboa.pt}) }
    \thanks{C. Silvestre and D. Cabecinhas are with the Department of Electrical and Computer Engineering, Faculty of Science and Technology, University of Macau, Taipa, Macau, China, and C. Silvestre is also on leave from IST, Portugal. ({\tt csilvestre@umac.mo} and {\tt dcabecinhas@umac.mo})}
}

%\markboth{IEEE TRANSACTIONS ON CONTROL SYSTEMS TECHNOLOGY,~Vol.~0, No.~0, Month~Year}%
%{Guerreiro \MakeLowercase{\textit{et al.}}: LiDAR-based Control of Autonomous Rotorcraft -- Supplementary Material}



\maketitle

%--- Abstract
\begin{abstract}
This is a complementary document to the paper presented in \cite{guerreiro:2017:tcst:lidarctrl}, where more detailed proofs are provided for some results. 
The main paper addresses the problem of trajectory tracking control of autonomous rotorcraft in operation scenarios where only relative position measurements obtained from LiDAR sensors are possible.
The proposed approach defines an alternative kinematic model, directly based on LiDAR measurements, and uses a trajectory-dependent error space to express the dynamic model of the vehicle.
An LPV representation with piecewise affine dependence on the parameters is adopted to describe the error dynamics over a set of predefined operating regions, and a continuous-time $\mathcal{H}_2$ control problem is solved using LMIs and implemented within the scope of gain-scheduling control theory.
In this document, Section \ref{lpv:app:il-stability} presents the stability analysis of the attitude inner-loop presented in \cite[Section II.B]{guerreiro:2017:tcst:lidarctrl}, whereas Section \ref{lpv:app:ctrl-synthesis} presents a more detailed version of the stability and performance guarantees for the LPV system, extending the results presented in \cite[Section IV.A]{guerreiro:2017:tcst:lidarctrl}.
\end{abstract}

%--- Appendices
\appendices
%\appendix

\section{Inner-loop Dynamics}
\label{lpv:app:il-stability}

This appendix presents a possible inner-loop stabilization method, within the framework of feedback linearization and Lyapunov stability methods, that results in a second-order linear model for the pitch and roll angular motion, as well as a first-order linear model for the yaw angular velocity.

\begin{thm}[Inner-loop Stability]
    \label{lpv:thm:innerloop}
    Consider the control law given by
    \begin{align}
        \bvec{n}\low{ext} = &\matskew{\om{}}\,\mat{J}\low{B}\,\om{} 
                    + \mat{J}\low{B}\,\mat{Q}\high{-1}(\lbd{}) [ 
                        \dot{\mat{Q}}(\lbd{})\,\om{} 
                        \!- \mat{K}_2\,( \dlbd{} - \bvec{e}_3\,\bvec{e}_3^T\,\ctr{IL} )\nonumber\\
                   &    \!- \gmat{\Pi}_{\bvec{e}_3}^T\,\mat{K}_1\,\gmat{\Pi}_{\bvec{e}_3}\,( \lbd{} - \ctr{IL} )
                    ]
            \label{model:eq:control-law}
    \end{align}
    where $\mat{K}_1 \in \R^2$ and $\mat{K}_2 \in \R^3$ are positive definite diagonal matrices and the inner-loop input vector is denoted as $\ctr{IL} = \blkmat{u_{\phi} & u_{\theta} & u_{\dot{\psi}}}^T$, accounting for the desired roll angle, pitch angle, and yaw angular rate, respectively. 
    Then, the resulting attitude dynamics is given by 
    \begin{align}
        \ddlbd{}   &= - \mat{K}_2\,( \dlbd{} - \bvec{e}_3 \bvec{e}_3^T \ctr{IL} )
                        - \gmat{\Pi}_{\bvec{e}_3}^T\,\mat{K}_1\,\gmat{\Pi}_{\bvec{e}_3}\,( \lbd{} - \ctr{IL} ) \;.
        \label{model:eq:dyn-il}
    \end{align}
%    is locally input-to-state stable 
%    and the state variables $\phi$, $\theta$, and $\dot{\psi}$ converge exponentially to the constant inputs $u_\phi$, $u_\theta$, and $u_{\dot{\psi}}$, respectively. 
    which, for constant inputs, guarantees that the equilibrium point $(\gmat{\Pi}_{\bvec{e}_3}\lbd{},\dlbd{}) = (\gmat{\Pi}_{\bvec{e}_3}\ctr{IL},\bvec{e}_3 \bvec{e}_3^T \ctr{IL})$ is exponentially stable.
\end{thm}
\begin{proof}
Recalling the rigid body dynamics, the angular motion equations can be written as
\begin{subequations}
\label{model:eq:angular-motion}
\begin{empheq}[left=\empheqlbrace]{align}
    \dom{}{} &= \mat{J}\low{B}^{-1} [ -\matskew{\om{}}\,\mat{J}\low{B} \om{} + \bvec{n}\low{ext} ]
        \label{model:eq:angular-dyn}        \\
    \dlbd{}   &= \mat{Q}(\lbd{})\,\om{} 
        \label{model:eq:angular-kin}
\end{empheq}
\end{subequations}
Noting that the angular velocity can be defined as $\om{} = \mat{Q}^{-1}(\lbd{}) \dlbd{}$ and taking the time derivative of angular kinematics, the angular motion dynamics can be expressed as
\begin{align}
    \ddlbd{}   %&= \dot{\mat{Q}}(\lbd{})\,\om{} + \mat{Q}(\lbd{})\,\dom{} \nonumber\\
               &= \dot{\mat{Q}}(\lbd{})\,\om{} 
                    + \mat{Q}(\lbd{})\,\mat{J}\low{B}^{-1} 
                        \left[  - \matskew{\om{}}\,\mat{J}\low{B}\,\om{} 
                                + \bvec{n}\low{ext} \right] %\\
%               &= \dot{\mat{Q}}(\lbd{}) \mat{Q}^{-1}(\lbd{}) \dlbd{}
%                    + \mat{Q}(\lbd{}) \mat{J}\low{B}^{-1} 
%       \left[  - \matskew{\mat{Q}^{-1}(\lbd{}) \dlbd{}} \mat{J}\low{B}}\mat{Q}^{-1}(\lbd) \dlbd{} 
%               + \bvec{n}_{ext} \right] \\
    \label{model:eq:il-angular-dyn}
\end{align}
Consider the desired roll angle, pitch angle, and yaw rate to be denoted as $u_{\phi} \in \R$, $u_{\theta} \in (-\pi/2,\pi/2)$, and $u_{\dot{\psi}} \in \R$, respectively, indicating that they will be considered as inputs in the overall system. 
Let also the input vector of this inner-loop system be defined as $\ctr{IL} = \blkmat{u_{\phi} & u_{\theta} & u_{\dot{\psi}}}^T$. 

Thus, a feedback linearizing control law can be designed as
\begin{align}
    \bvec{n}\low{ext} = &\matskew{\om{}}\,\mat{J}\low{B}\,\om{} 
                + \mat{J}\low{B}\,\mat{Q}\high{-1}(\lbd{}) [ 
                    \dot{\mat{Q}}(\lbd{})\,\om{} 
                    \!- \mat{K}_2\,( \dlbd{} - \bvec{e}_3\,\bvec{e}_3^T\,\ctr{IL} )\nonumber\\
               &    \!- \gmat{\Pi}_{\bvec{e}_3}^T\,\mat{K}_1\,\gmat{\Pi}_{\bvec{e}_3}\,( \lbd{} - \ctr{IL} )
                ]
        \label{model:eq:il-control-law}
\end{align}
where $\mat{K}_1 \in \R^{2\times2}$, $\mat{K}_2 \in \R^{3\times3}$, $\bvec{e}_3 = \blkmat{0 & 0 & 1}^T$, and as such $\bvec{e}_3 \bvec{e}_3^T \ctr{IL} = \blkmat{0 & 0 & u_{\dot{\psi}}}^T$, while $\gmat{\Pi}_{\bvec{e}_3} = \blkmat{\eye{2} & \zeros{2}{1}}$ such that $\gmat{\Pi}_{\bvec{e}_3} \ctr{IL} = \blkmat{ u_{\phi} & u_{\theta}}^T$. 
Replacing the control law \eqref{model:eq:il-control-law} in the angular dynamics \eqref{model:eq:il-angular-dyn}, it can be seen that the closed-loop dynamics results in 
\begin{align}
    \ddlbd{}   &= - \mat{K}_2\,( \dlbd{} - \bvec{e}_3\,\bvec{e}_3^T\,\ctr{IL} )
                    - \gmat{\Pi}_{\bvec{e}_3}^T\,\mat{K}_1\,\gmat{\Pi}_{\bvec{e}_3}\,( \lbd{} - \ctr{IL} ) \;.
    \label{model:eq:il-angular-closed-loop}
\end{align}
Considering that $\ctr{IL} \in \mathcal{U}\low{IL} \subset \R^3$ and define a state vector as $\bvec{x}\low{IL} = \blkmat{ \gmat{\Pi}_{\bvec{e}_3}\lbd{} & \dlbd{} }^T \in \mathcal{X}\low{IL}$, the resulting system is a linear time-invariant (LTI) system of the form 
\[
    \dot{\bvec{x}}\low{IL} = \mat{A}\low{IL}\,\bvec{x}\low{IL} + \mat{B}\low{IL}\,\ctr{IL}\;,
\]
where the system matrices are given by 
\begin{align*}
    \mat{A}\low{IL}   &= \blkmat{ \zeros{2}{2}                      & \gmat{\Pi}_{\bvec{e}_3} \\
                                  -\gmat{\Pi}_{\bvec{e}_3}^T\,\mat{K}_1    & -\mat{K}_2        } \;, \\
    \mat{B}\low{IL}   &= \blkmat{ \zeros{2}{3}                                           \\
                                  \gmat{\Pi}_{\bvec{e}_3}^T\,\mat{K}_1\,\gmat{\Pi}_{\bvec{e}_3}
                                    +\mat{K}_2\,\bvec{e}_3\,\bvec{e}_3^T        }\;.
\end{align*}
It can be seen that matrix $\mat{A}\low{IL}$ is Hurwitz for any positive definite matrices $\mat{K}_1$ and $\mat{K}_2$, and, therefore, the system is locally input-to-state stable, noting that the local part of the result is a direct consequence of the domain of the Euler angles not being $\R^3$.

Considering the matrices $\mat{K}_1 = \diag{k_\phi,k_\theta}$ and $\mat{K}_2 = \diag{k_{\dot{\phi}},k_{\dot{\theta}},k_{\dot{\psi}}}$ the system described in \eqref{model:eq:il-angular-closed-loop} can also be written as 
\begin{subequations}
\label{model:eq:il-simple}
\begin{empheq}[left=\empheqlbrace]{align}
    \ddot{\phi}     &= -k_{\dot{\phi}}\,\dot{\phi} - k_\phi\,(\phi - u_\phi )
                    \\
    \ddot{\theta}   &= -k_{\dot{\theta}}\,\dot{\theta} - k_\theta\,(\theta - u_\theta )\;.
                    \\
    \ddot{\psi}     &= -k_{\dot{\psi}}\,(\dot{\psi} - u_{\dot{\psi}} )
\end{empheq} 
\end{subequations}
It can be seen that, with constant inputs and the change of variables $\tilde{\phi} = \phi - u_\phi$, $\tilde{\theta} = \theta - u_\theta$, and $\dot{\tilde{\psi}} = \dot{\psi} - u_{\dot{\psi}}$, the resulting autonomous system is exponentially stable.
Thus, it can be concluded that the state variables $\phi$, $\theta$, and $\dot{\psi}$ converge exponentially to the constant inputs $u_\phi$, $u_\theta$, and $u_{\dot{\psi}}$, respectively.

\end{proof}

\section{Controller Synthesis}
\label{lpv:app:ctrl-synthesis}

In this section an LMI approach is used to tackle the continuous-time state feedback $\mathcal{H}_2$ synthesis problem for polytopic LPV systems. Consider a general LPV system of the form
\begin{subequations}
\label{lpv:eq:ctr-sys}
\begin{empheq}[left=\empheqlbrace]{align}
        \dot{\bvec{x}}    &=   \mat{A}(\gvec{\xi})\bvec{x} + 
                                \mat{B}_w(\gvec{\xi})\bvec{w} + 
                                \mat{B}(\gvec{\xi})\bvec{u} \\
        \bvec{z}          &=   \mat{C}(\gvec{\xi})\bvec{x} + 
                                \mat{D}(\gvec{\xi})\bvec{w} + 
                                \mat{E}(\gvec{\xi})\bvec{u} 
\end{empheq}
\end{subequations}
where $\bvec{x}$ is the state, $\bvec{u}$ is the control input, $\bvec{z}$
denotes the error signal to be controlled, and $\bvec{w}$ denotes the
exogenous input signal. 
The system is parameterized by $\gvec{\xi}$, which is a possibly time-varying 
parameter vector and belongs to the convex set $\set{E}^j = \textrm{co}(\set{E}^j_0)$.
Here, the operator $\textrm{co}(.)$ denotes the convex hull of the elements of the 
argument set, $\set{E}^j_0 = \{\gvec{\xi}_1,\ldots,\gvec{\xi}_{n_j}\}$, where $\gvec{\xi}_1$ 
to $\gvec{\xi}_{n_j}$ are the vertices of a polytope. 
It is also noted that the controller synthesis presented in the following subsection will only be valid 
for a specific operating region, here represented by $\set{E}^j \subset \set{E}$. 

Applying the static state feedback law given by $\bvec{u} = \mat{K}\,\bvec{x}$ 
to \eqref{lpv:eq:ctr-sys} results in the closed-loop system given by
\begin{subequations}
\label{lpv:eq:ctr-sys-cl}
\begin{empheq}[left={\mat{T}_{zw}(\gvec{\xi}):=\empheqlbrace}]{align}
        \dot{\bvec{x}}    &= \mat{A}_c(\gvec{\xi})\,\bvec{x} + \mat{B}_c(\gvec{\xi})\,\bvec{w} \\
        \bvec{z}          &= \mat{C}_c(\gvec{\xi})\,\bvec{x} + \mat{D}_c(\gvec{\xi})\,\bvec{w} 
\end{empheq}
\end{subequations}
where $\mat{T}_{zw}(\gvec{\xi})$ denotes the resulting closed-loop operator from 
the disturbance input $\bvec{w}$ to the performance output $\bvec{z}$, and
the system matrices are defined as
$\mat{A}_c(\gvec{\xi}) = \mat{A}(\gvec{\xi})+\mat{B}(\gvec{\xi})\,\mat{K}$, 
$\mat{B}_c(\gvec{\xi}) = \mat{B}_w(\gvec{\xi})$, 
$\mat{C}_c(\gvec{\xi}) = \mat{C}(\gvec{\xi}) + \mat{E}(\gvec{\xi})\,\mat{K}$ 
and $\mat{D}_c(\gvec{\xi}) = \mat{D}(\gvec{\xi})$. 
The closed-loop system can be characterized in terms of quadratic stability 
using the following definition, where $\mat{X} \succ \zero$ denotes that matrix $\mat{X}$ is positive definite. 
\begin{define}[Quadratic stability, \cite{boyd:1994}]
    The system  is said to be quadratically stable if there exists a matrix 
    $\mat{X} \succ \zero$ such that 
    $\mat{A}_c^T(\gvec{\xi})\,\mat{X} + \mat{X}\,\mat{A}_c(\gvec{\xi}) \prec \zero$
    is satisfied for all $\gvec{\xi} \in \set{E}^j$.
\end{define}\noindent
It can be seen that testing for stability or solving the synthesis problem without any further result, involves an infinite number of LMIs. 
Thus, several different structures for LPV systems have been proposed which reduce the problem to that of solving a finite number of LMIs. 

In this section, an affine polytopic description is adopted, which can also be used to model a wide spectrum of systems and, as shown in the results presented in \cite{guerreiro:2016:tcst:lidarctrl}, is an adequate choice for the system at hand.
\begin{define}[Affine polytopic LPV system] \label{def:LPV}
    The system \eqref{lpv:eq:ctr-sys} is said to be a polytopic LPV system if
    the system matrix
    \begin{equation}
        \mat{P}(\gvec{\xi}) = \begin{bmatrix}
            \mat{A}(\gvec{\xi}) & \mat{B}_w(\gvec{\xi}) & \mat{B}(\gvec{\xi}) \\ 
            \mat{C}(\gvec{\xi}) & \mat{D}(\gvec{\xi}) & \mat{E}(\gvec{\xi}) \\
        \end{bmatrix}
    \end{equation}
    verifies $\mat{P}(\gvec{\xi}) \in \textrm{co}\left(\mat{P}_1,\ldots,\mat{P}_{n_r}\right)$
    for all $\gvec{\xi} \in \set{E}^j$, where
    \begin{equation*}
        \mat{P}_i = \begin{bmatrix}
            \mat{A}_i & \mat{B}_{w_i} & \mat{B}_i \\ 
            \mat{C}_i & \mat{D}_i & \mat{E}_i \\
        \end{bmatrix}
    \end{equation*}
    for all $i = 1,\ldots,n_j$. 
    Moreover, if $\set{E}^j$ is a polytopic set, such as 
    $\set{E}^j  = \convhull{\set{E}^j_0}$, $\set{E}^j_0 
                = \{\gvec{\xi}_1,\ldots,\gvec{\xi}_{n_j}\}$, and $\mat{P}(\gvec{\xi})$
    depends affinely on $\gvec{\xi}$, then $\mat{P}_i = \mat{P}(\gvec{\xi}_i)$ 
    for all $i = 1,\ldots,n_j$, i.e., the vertices of the parameter set can be 
    uniquely identified with the vertices of the system.
\end{define}\noindent
This polytopic structure used with the following lemma, enables the use of a
powerful set of results. 
\begin{prop}[{\cite[Proposition 1.19]{scherer:2000}}]
    \label{lpv:prop:polytopic-func}
    Let $\bvec{f}: \set{E}^j \rightarrow \R$ be a convex function defined on the convex
    set $\set{E}^j = \convhull{\set{E}^j_0}$. Then, for some $\gamma \in \R$, 
    $\bvec{f}(\gvec{\xi}) \leq \gamma$ for all $\gvec{\xi} \in \set{E}^j$ if and only if 
    $\bvec{f}(\gvec{\xi}) \leq \gamma$ for all $\gvec{\xi} \in \set{E}^j_0$.
\end{prop}\noindent
Thus, the quadratic stability of an affine polytopic LPV system can be easily 
established if there exists a matrix $\mat{X} \succ \zero$ such that 
$\mat{A}_c^T(\gvec{\xi})\,\mat{X} + \mat{X}\,\mat{A}_c(\gvec{\xi}) \prec \zero$
is satisfied for all $\gvec{\xi} \in \set{E}^j_0$. 

The $\mathcal{H}_2$ synthesis problem can be described as that of
finding a control matrix $\mat{K}$ that stabilizes the closed-loop system
and minimizes the $\mathcal{H}_2$-norm of $\mat{T}_{zw}(\gvec{\xi})$, 
denoted by $\normH{2}{\mat{T}_{zw}(\gvec{\xi})}$. 
It is assumed that matrix $\mat{D}(\gvec{\xi}) = 0$ in order to guarantee that $\normH{2}{\mat{T}_{zw}(\gvec{\xi})}$ is finite for every internally stabilizing and strictly proper controller. 
The following theorem is used for controller design and relies on results 
available in \cite{ghaoui:1999} and \cite{scherer:2000}, after being rewritten 
for the case of polytopic LPV systems. 
In the following, $\tr{.}$ denotes the trace of the argument matrix.
\begin{thm}[Polytopic stability]
    \label{lpv:thm:LMIsynthesis}
    If there are real matrices $\mat{X} = \mat{X}^T \succ 0$, $\mat{Y} \succ 0$, 
    and $\mat{W}$ such that
    \begin{subequations}
    \label{lpv:eq:ctr-LMIsyn-full}
    \begin{align}
        { \blkmat{
            \mat{A}(\gvec{\xi}) \mat{X} 
            \!\!+\!\! \mat{X} \mat{A}^T\!(\gvec{\xi}) 
            \!+\! \mat{B}(\gvec{\xi}) \mat{W} 
            \!\!+\!\! \mat{W}^T \mat{B}^T\!(\gvec{\xi}) 
                & \!\!\!\mat{B}_w(\gvec{\xi})    \\
            \mat{B}_w^T(\gvec{\xi})  
                & \!\!\!-\eyefunc                \\
        } \!\!\prec  \!0 }
        \label{lpv:eq:ctr-LMIsyn-1}
        \\
        { \blkmat{
            \mat{Y}         & \mat{C}(\gvec{\xi})\,\mat{X} + \mat{E}(\gvec{\xi})\,\mat{W}   \\
            \mat{X}\,\mat{C}^T(\gvec{\xi}) + \mat{W}^T\,\mat{E}^T(\gvec{\xi})   & \mat{X}   \\
        } \!\!\succ \!0 }
        \label{lpv:eq:ctr-LMIsyn-2} 
        \\
        { \tr{\mat{Y}} \!\!<\! \gamma^2 }
        \label{lpv:eq:ctr-LMIsyn-3}
    \end{align}
    %\noeqref{lpv:eq:ctr-LMIsyn-full,lpv:eq:ctr-LMIsyn-1,lpv:eq:ctr-LMIsyn-2,lpv:eq:ctr-LMIsyn-3}
    \end{subequations}
    for all $\gvec{\xi} \in \set{E}^j_0$, where the static feedback controller is 
    defined as $\mat{K} = \mat{W}\,\mat{X}^{-1}$, then, the closed-loop system 
    is quadratically stable and there exists an upper-bound $\gamma$ for the 
    continuous-time $\mathcal{H}_2$-norm of the 
    closed-loop operator $\mat{T}_{zw}(\gvec{\xi})$ for all $\gvec{\xi} \in \set{E}^j$, 
    i.e.,
    \begin{equation}\label{lpv:eq:ctr-LMIsyn-4}
        \normH{2}{\mat{T}_{zw}(\gvec{\xi})} < \gamma\;, \forall\gvec{\xi} \in \set{E}^j\;.
    \end{equation}
\end{thm}
\begin{proof}
Using Proposition \ref{lpv:prop:polytopic-func} and assuming an affine polytopic LPV system, it can be seen that satisfying the LMI system for all $\gvec{\xi} \in \set{E}^j_0$ is equivalent to satisfying the same system for all $\gvec{\xi} \in \set{E}^j$.
The proof that the LMI system \eqref{lpv:eq:ctr-LMIsyn-full} implies \eqref{lpv:eq:ctr-LMIsyn-4} can be obtained from the definition of $\mathcal{H}_2$-norm of $\mat{T}_{zw}(\gvec{\xi})$.
Let the transfer function matrix of the closed-loop operator $\mat{T}_{zw}(\gvec{\xi})$ be denoted as $\mat{T}_{i_{zw}}(s)$ for some parameter vector $\gvec{\xi}_i \in \set{E}^j_0$, and defined as
\begin{align*}
    &\mat{T}_{i_{zw}}(s) = \\
    &= \mat{C}_c(\gvec{\xi}_i)\,
        \left[s\,\eyefunc - \mat{A}_c(\gvec{\xi}_i)\right]^{-1}\,
        \mat{B}_c(\gvec{\xi}_i) + \mat{D}_c(\gvec{\xi}_i) \\
    &= \left[\mat{C}(\gvec{\xi}_i) + \mat{E}(\gvec{\xi}_i)\,\mat{K}\right]\,
            \left[s\,\eyefunc - \mat{A}(\gvec{\xi}_i)-\mat{B}(\gvec{\xi}_i)\,\mat{K}\right]^{-1}\,
            \mat{B}_w(\gvec{\xi}_i) \\
    &   \quad + \mat{D}(\gvec{\xi}_i)
    \;.
\end{align*}
Then, the $\mathcal{H}_2$-norm of $\mat{T}_{i_{zw}}(s)$
is defined as
\begin{align*}
    \normH{2}{\mat{T}_{i_{zw}}}^2 
        &= \frac{1}{2\,\pi} \tr{\int_{-\infty}^{+\infty} 
                \mat{T}_{i_{zw}}^H(j\omega)\,\mat{T}_{i_{zw}}(j\omega)\,d\omega}
\intertext{which using Parseval's theorem can be rewritten as}
    \normH{2}{\mat{T}_{i_{zw}}}^2 
        &= \tr{\int_{0}^{+\infty} 
            \mat{H}_i^T(t)\,\mat{H}_i(t)\,dt}
\end{align*}
where $\mat{H}_i(t)$ is the impulse response matrix of 
$\mat{T}_{i_{zw}}(s)$. 
Noting that the impulse response can be defined as
\[
    \mat{H}_i(t) = \mat{C}_c(\gvec{\xi}_i)\,e^{\mat{A}_c(\gvec{\xi}_i)\,t}\,\mat{B}_c(\gvec{\xi}_i)\;,
\]
it is possible to see that, after some algebraic manipulation,
\begin{align*}
    \normH{2}{\mat{T}_{i_{zw}}}^2 
%            &= \tr{ \int_{0}^{+\infty} 
%                        \mat{B}_c^T(\gvec{\xi}_i)\,e^{\mat{A}_c^T(\gvec{\xi}_i)\,t}\,\mat{C}_c^T(\gvec{\xi}_i)\,
%                        \mat{C}_c(\gvec{\xi}_i)\,e^{\mat{A}_c(\gvec{\xi}_i)\,t}\,\mat{B}_c(\gvec{\xi}_i)\,dt } \\
%            &= \tr{ \int_{0}^{+\infty} 
%                        \mat{C}_c(\gvec{\xi}_i)\,e^{\mat{A}_c(\gvec{\xi}_i)\,t}\,\mat{B}_c(\gvec{\xi}_i)\,
%                        \mat{B}_c^T(\gvec{\xi}_i)\,e^{\mat{A}_c^T(\gvec{\xi}_i)\,t}\,\mat{C}_c^T(\gvec{\xi}_i)\,dt } \\
%            &= \tr{ \mat{C}_c(\gvec{\xi}_i)\,\int_{0}^{+\infty} 
%                        e^{\mat{A}_c(\gvec{\xi}_i)\,t}\,\mat{B}_c(\gvec{\xi}_i)\,\mat{B}_c^T(\gvec{\xi}_i)\,e^{\mat{A}_c^T(\gvec{\xi}_i)\,t}\,dt\,\mat{C}_c^T(\gvec{\xi}_i) } \\
        &= \tr{ \mat{C}_c(\gvec{\xi}_i)\,\mat{W}_{ctr}(\gvec{\xi}_i)\,\mat{C}_c^T(\gvec{\xi}_i) }\;,
\end{align*}
% \intertext{where}
where
\begin{align*}
    \mat{W}_{ctr}(\gvec{\xi}_i) 
        &= \int_{0}^{+\infty} 
            e^{\mat{A}_c(\gvec{\xi}_i)\,t}\,\mat{B}_c(\gvec{\xi}_i)\,
                \mat{B}_c^T(\gvec{\xi}_i)\,e^{\mat{A}_c^T(\gvec{\xi}_i)\,t}\,dt
\end{align*}
stands for the controllability Grammian, given by the symmetric positive 
definite solution of the Lyapunov equation
\begin{equation*}
    \mat{A}_c(\gvec{\xi}_i)\,\mat{W}_{ctr}(\gvec{\xi}_i) + \mat{W}_{ctr}(\gvec{\xi}_i)\,\mat{A}_c^T(\gvec{\xi}_i) + \mat{B}_c(\gvec{\xi}_i)\,\mat{B}_c^T(\gvec{\xi}_i) = \zero\;.
\end{equation*}
It can be seen that, for each parameter vector $\gvec{\xi}_i \in \set{E}^j_0$, 
$\gamma$ is an upper bound for the $\mathcal{H}_2$-norm 
of the closed-loop operator $\mat{T}_{i_{zw}}$ if and only if there exists 
$\mat{X} \succ 0$, $0 \prec \mat{W}_{ctr}(\gvec{\xi}_i) \prec \mat{X}$ such that
\begin{equation}\label{lpv:eq:ctr-LMIproof-1}
    \mat{A}_c(\gvec{\xi}_i)\,\mat{X} + \mat{X}\,\mat{A}_c(\gvec{\xi}_i)^T + \mat{B}_c(\gvec{\xi}_i)\,\mat{B}_c^T(\gvec{\xi}_i) \prec \zero
\end{equation}
and
\begin{equation}\label{lpv:eq:ctr-LMIproof-2}
    \tr{\mat{C}_c(\gvec{\xi}_i)\,\mat{X}\,\mat{C}_c^T(\gvec{\xi}_i)} < \gamma^2.
\end{equation}
Equation \eqref{lpv:eq:ctr-LMIproof-1} can be rewritten as
\begin{equation*}
    \blkmat{
        \mat{A}_c(\gvec{\xi}_i)\,\mat{X} 
            + \mat{X}\,\mat{A}_c^T(\gvec{\xi}_i) 
            + \mat{B}_c(\gvec{\xi}_i)\,\mat{B}_c^T(\gvec{\xi}_i)  
          & \bvec{0}    \\
        \bvec{0}                          
          & -\eyefunc        \\
    } \prec \zero \;,
\end{equation*}
that, using Schur complements, becomes
\begin{equation*}
    \blkmat{
        \mat{A}_c(\gvec{\xi}_i)\,\mat{X} 
            + \mat{X}\,\mat{A}_c^T(\gvec{\xi}_i)   
          & \mat{X}\,\mat{B}_c(\gvec{\xi}_i)   \\
        \mat{B}_c^T(\gvec{\xi}_i)\,\mat{X}               
          & -\eyefunc    \\
    } \prec \zero \;,
\end{equation*}
or equivalently, introducing the matrix $\mat{W} = \mat{K}\,\mat{X}$ yields
\begin{equation}\label{lpv:eq:ctr-LMIproof-1-3}
    {\scriptstyle \blkmat{
        \mat{A}(\gvec{\xi}_i) \mat{X} 
            \!\!+\!\! \mat{X} \mat{A}^T\!(\gvec{\xi}_i) 
            \!+\! \mat{B}(\gvec{\xi}_i) \mat{W} 
            \!\!+\!\! \mat{W}^T \mat{B}^T\!(\gvec{\xi}_i) 
          & \!\!\!\mat{B}(\gvec{\xi}_i)     \\
        \mat{B}^T(\gvec{\xi}_i)                             
          & \!\!\!-\eyefunc    \\
    } \!\!\prec \zero. }
\end{equation}
%This last LMI matches \eqref{lpv:eq:ctr-LMIsyn-1} and, as the conditions 
%of the theorem state that there exists $\mat{X} = \mat{X}^T \succ 0$ and 
%$\mat{W}$ for which this LMI is satisfied for all $\gvec{\xi}_i \in \set{E}^j_0$,
%then
Under the conditions of the theorem, it can be seen that 
\eqref{lpv:eq:ctr-LMIproof-1-3} is satisfied for all 
$\gvec{\xi}_i \in \set{E}^j_0$ and that \eqref{lpv:eq:ctr-LMIsyn-2} can be 
written as
\begin{equation*}
    \blkmat{
        \mat{Y}       & \mat{C}_c(\gvec{\xi}_i)       \\
        \mat{C}_c^T(\gvec{\xi}_i)   & \mat{X}^{-1}    \\
    } \succ \zero \;,
\end{equation*}
which, applying once again Schur complements, implies that 
\begin{align*}
    \blkmat{
        \mat{Y} - \mat{C}_c(\gvec{\xi}_i)\,\mat{X}\,\mat{C}_c^T(\gvec{\xi}_i)   & \bvec{0} \\
        \bvec{0}              & \mat{X}^{-1} \\
    } &\succ \zero
\intertext{and consequently}
    \tr{\mat{C}_c(\gvec{\xi}_i)\,\mat{X}\,\mat{C}_c(\gvec{\xi}_i)^T} < \tr{\mat{Y}} &< \gamma^2\;,
\end{align*}
also, for all $\gvec{\xi}_i \in \set{E}^j_0$.
Thus, \eqref{lpv:eq:ctr-LMIproof-2} is implied and the bound on the 
$\set{H}_2$-norm is established.
As \eqref{lpv:eq:ctr-LMIproof-1} is implied by the conditions of the theorem,
it can be noted that
\[
    \mat{A}_c(\gvec{\xi}_i)\,\mat{X} + \mat{X}\,\mat{A}_c(\gvec{\xi}_i)^T 
    \prec -\mat{B}_c(\gvec{\xi}_i)\,\mat{B}_c^T(\gvec{\xi}_i) 
    \preceq \zero \;,
\]
for all $\gvec{\xi}_i \in \set{E}^j_0$, implying the quadratic stability of the 
closed-loop and, thus, concluding the proof.
    
\end{proof}\noindent
With this result, the optimal solution for the continuous-time $\mathcal{H}_2$ control problem is approximated through the minimization of $\gamma$ subject to the LMIs of Theorem \ref{lpv:thm:LMIsynthesis}.

% Bibliography
\bibliographystyle{IEEEtran}
\bibliography{IEEEabrv,Biblios_Lidar,Biblios_Self,Biblios_Intro,Biblios_DSOR,Biblios_Models}

\end{document}